%% file: main.tex
\documentclass[a4paper, USenglish, thm-restate]{lipics-v2021}

\hideLIPIcs
\nolinenumbers

\usepackage[utf8]{inputenc}
\usepackage[T1]{fontenc}
\usepackage{needspace}
\usepackage{comment} 
\usepackage{libertine} 
\usepackage{inconsolata} 

\usepackage{amsmath, amsthm, amssymb}
\usepackage{stmaryrd}
\usepackage{graphicx}
\usepackage{enumerate}
\usepackage{thmtools}
\usepackage{doi} 
\usepackage{tikz}
\usepackage[numbers, sort&compress]{natbib} 

\declaretheorem[name=Question, style=remark, sibling=theorem]{question}

\title{Local certification of MSO properties for bounded treedepth graphs}

\author{Nicolas Bousquet}{Univ. Lyon, Université Lyon 1, LIRIS UMR CNRS 5205, F-69621, Lyon, France}{nicolas.bousquet@univ-lyon1.fr}{https://orcid.org/0000-0003-0170-0503}{}

\author{Laurent Feuilloley}{Univ. Lyon, Université Lyon 1, LIRIS UMR CNRS 5205, F-69621, Lyon, France}{laurent.feuilloley@univ-lyon1.fr}{https://orcid.org/0000-0002-3994-0898}{}

\author{Théo Pierron}{Univ. Lyon, Université Lyon 1, LIRIS UMR CNRS 5205, F-69621, Lyon, France}{theo.pierron@univ-lyon1.fr}{https://orcid.org/0000-0002-5586-5613}{}

\authorrunning{N. Bousquet, L. Feuilloley and T. Pierron}

\funding{This work was supported by ANR project GrR (ANR-18-CE40-0032)}



\begin{document}

\maketitle
\thispagestyle{empty}

\begin{abstract}
The graph model checking problem consists in testing whether an input graph satisfies a given logical formula. In this paper, we study this problem in a distributed setting, namely local certification. The goal is to assign labels to the nodes of a network to certify that some given property is satisfied, in such a way that the labels can be checked locally.

We first investigate which properties can be locally certified with small certificates. Not surprisingly, this is almost never the case, except for not very expressive logic fragments. Following the steps of Courcelle-Grohe, we then look for meta-theorems explaining what happens when we parameterize the problem by some standard measures of how simple the graph classes are. In that direction, our main result states that any MSO formula can be locally certified on graphs with bounded treedepth with a logarithmic number of bits per node, which is the golden standard in certification.
\end{abstract}

\bigskip

\thispagestyle{empty}
\clearpage
\pagenumbering{arabic} 

\section{Introduction}
\subsection{Local certification}

In this work, we are interested in the locality of graph properties. For example, consider the property ``the graph has maximum degree three''. 
We say that this property can be checked locally, because if every node checks that it has at most three neighbors (which is a local verification), then the graph satisfies the property (which is a global statement).
Most graph properties of interest are not local. 
For example, to decide whether a graph is acyclic, or planar, the vertices would have to look further in the graph. Some properties can be seen as local or not, depending on the exact definition. For example, having a diameter at most $3$, is a property that can be checked locally if we consider that looking at distance $3$ is local, but not if we insist on inspecting only the neighbors of a vertex. 

In general, checking whether a network has some given property is an important primitive for graph algorithms, as it allows to choose the most efficient algorithms for that network. For distributed graph algorithms in particular, locality is a guarantee of efficiency. Therefore, properties that can be checked locally are especially useful.  
As said, most properties are not locally checkable, and we would like to have a mechanism to circumvent this shortcoming.

Local certification is such a mechanism, in the sense  that it allows to check any graph property locally. 
For a given property, a local certification is described by a certificate assignment and a verification algorithm: each node receives a certificate, reads the certificates of its neighbors and then runs a verification algorithm~\cite{Feuilloley21}. This algorithm decides whether the node accepts or rejects the certification. 
If the graph satisfies the property, then there should be a certificate assignment such that all the nodes accept. Otherwise, in each assignment, there must be at least one node that rejects. Note that since we consider a distributed setting, we work under the standard assumptions that the graphs we consider are connected, loopless and non-empty.

For concreteness, let us describe a local certification of acyclicity~\cite{AfekKY90}. 
On an acyclic graph, we assign the certificates in the following way: we choose a node to root the tree, and then each certificate consists of the distance to the root.
Given those certificates, the vertices can check that the distances are sound. 
The key point is that if the graph has a cycle, then in every cycle, for any certificate assignment, there will be an inconsistency in the distances, and then at least one vertex will reject.

Note that the definition of local certification is similar to the certificate definition of the complexity class NP.
And indeed, it can be seen as a form of non-determinism for distributed decision~\cite{GoosS16, KormanKP10, FraigniaudKP13}. 
Actually, local certification does not originate from the theoretical study of local properties and non-determinism. Instead, it appeared in the context of self-stabilization in distributed computing, where an algorithm, in addition to computing a solution to a combinatorial problem (\emph{e.g.} building a spanning tree), would also compute a certification of it, in order to check its correctness locally, in case of faults (for example a change of pointer in a tree) \cite{KormanKP10, AfekKY90, AwerbuchPVD94}.

Any property can be certified with $O(n^2)$ bits certificates, where $n$ is the total number of vertices. 
This is because one can just give the full description of the graph to every node, which can then check that the property holds in the description, that the description is correct locally, and identical between neighbors. 
This $O(n^2)$ size is very large and not practical, thus the main goal of the study of local certification, is to minimize the size of the certificate, expressed as a number of bits per vertex, as a function of $n$.
In addition to the optimization motivation originating from distributed self-stabilizing algorithms, establishing the minimum certificate size also has a more theoretical appeal. Indeed, the optimal certification size of a property can be seen as a measure of its locality: the smaller the labels, the less global information we need to provide to allow local verification, the more local the property.

For acyclicity, the optimal certificate size is $\Theta(\log n)$ bits. In the certification described above, the distances are between $0$ and $n$, thus can be encoded in $O(\log n)$ bits, and there is a matching lower bound~\cite{KormanKP10, GoosS16}.
For most properties, one cannot hope to go below this $\Omega(\log n)$ bits per node. This size is needed for distances, but also to store identifiers, that are necessary to break symmetry for some properties (we consider that each vertex has a unique identifier represented by an integer from $[1,n^k]$ where $k$ is a fixed constant). As a consequence, $\Theta(\log n)$ has been identified as the standard for compact (or efficient) certification. 
Recently, planarity and more generally embeddability on bounded-genus surfaces, and $H$-minor-freeness for specific $H$, have been proved to have such compact certifications~\cite{FeuilloleyFMRRT21, FeuilloleyFMRRT21-b, EsperetL21, BousquetFP21}.

Unfortunately, not every property has a compact certification.
For example, having a non-trivial automorphism or not being 3-colorable are properties that basically cannot be certified with less than $\Omega(n^2)$ bits~\cite{GoosS16}.
Even surprisingly simple properties, such as having diameter at most $2$, cannot be certified with a sublinear number of bits per vertex (up to logarithmic factors), if we only allow the local verification to be at distance~one~\cite{Censor-HillelPP20}.

\subsection{Graph model checking}

As mentioned above, many specific graph properties such as planarity or small-diameter have been studied in the context of local certification. 
In this paper, we take a more systematic approach, by considering classes of graph properties. 
We are interested in establishing theorems of the form: ``all the properties that can be expressed in some formalism $X$ have a compact certification''. 

Considering such an abstraction already proved its worth in sequential computing: instead of proving results for specific graph properties, one can extract the common ideas behind their proof and obtain metatheorems. Those theorems state links between some classes of logical sentences\footnote{Some readers might be more familiar with the word ``formula'' than with ``sentence''. In our context, ``sentence'' is more correct, but for the introduction, one could harmlessly replace the latter by the former.} and algorithmic properties such as fixed-parameter tractability~\cite{Courcelle90,Grohe17,Dvorak10} or approximability~\cite{Papadimitriou91}, see for example the survey~\cite{kreutzer11}. The stereotypical example is Courcelle's result~\cite{Courcelle90}, which shows that model checking of MSO sentences is FPT when parameterized by the treewidth of the input graph. Observe that for these results, there is a tradeoff between the expressiveness of the fragment and the generality of the conclusion. For example, when considering first order logic (i.e. forbidding quantifications on sets), the model checking becomes FPT on nowhere dense classes~\cite{Grohe17}.

In this paper, we will consider properties that can be expressed by sentences from monadic second order logic (MSO), just like in Courcelle's theorem. 
These are formed from atomic predicates that test equality or adjacency of vertices (recall that our graphs are loopless hence equality implies non-adjacency) and allowing boolean operations and quantifications on vertices, edges, and sets of vertices or edges. Now, certifying a given property consists in certifying that a graph is a positive instance of the so-called graph model checking problem for the corresponding sentence $\varphi$: 
\begin{itemize} 
    \item Input: A graph $G$
    \item Output: Yes, if and only if, $G$ satisfies $\varphi$.
\end{itemize}

\subsection{The generic case}

Let us first discuss what such a metatheorem must look like when we do not restrict the class of graphs we consider. As we already mentioned, diameter $2$ graphs cannot be certified with sublinear certificates~\cite{Censor-HillelPP20}. 
This can be expressed with the following sentence:
$$\forall x\forall y (x=y\vee x-y \vee \exists z (x-z \land z-y))$$
This sentence is very simple: it is a first order sentence (a special case of MSO), it has quantifier depth three and there is only one quantifier alternation (two standard complexity measures for FO sentences which respectively counts the maximum number of nested quantifiers and the number of alternations between blocks of existential and universal quantifiers). 

Therefore, there exists very simple first order logic sentences which cannot be certified efficiently. The only possible way to simplify the sentences would consist in only having at most two nested quantifiers or not authorizing alternation of quantifiers. In these cases, the following holds:

\begin{restatable}{lemma}{lemFragments}
\label{lem:existsFO}
FO sentences with quantifier depth at most 2 can be certified with $O(\log n)$ bits. \\
Existential FO sentences (\emph{i.e.} whose prenex normal form has only existential quantifiers) can be certified with $O(\log n)$ bits.
\end{restatable}

Section~\ref{sec:fragment} is devoted to the easy proof of this result. We sketch it here. For the FO sentences with quantifiers of depth at most $2$, we can prove that the only interesting properties that can be expressed are basically a vertex being dominant (adjacent to all other vertices) or the graph being a clique. As it will be described in  Subsection~\ref{subsec:basic-certif}, we can easily certify the number of vertices in the graph. Given this information, a given vertex (resp. all vertices) can check that its degree is (resp. their degree are) equal to $n-1$, which is sufficient for these properties. For existential FO sentences, the certification boils down to having a spanning tree pointing to each vertex used in a variable, and certifying the structure of the subgraph induced by these vertices (by simply describing it completely in each certificate, which takes $O(k^2+k\log n)$ bits).

One can notice that we do not mention in this lemma the formulas which only contain universal ($\forall$) quantifiers. Actually, certifying these properties seem to be much harder than existential formulas. Indeed, it is, for instance, much easier to exhibit a triangle if it exists (existential FO: $\exists x\exists y\exists z (x-y\wedge y-z\wedge x-z)$) than proving that the graph is triangle-free since we can pinpoint the vertices mapped to the existentially quantified variables. We actually conjecture that the answer to the following question is negative:

\begin{question}
Is it possible to certify with $O(\log n)$ bits that a graph is triangle-free?
\end{question}

The main reason for our conjecture is that this question has a similar flavor as the triangle detection problem in the CONGEST model. 
The problem is basically the same, in the sense that at the end at least one node should have a special output if and only if there is a triangle. But the model is different: in the CONGEST model the vertices exchange messages of size $O(\log n)$, and we are interested in the number of rounds of communication before the vertices output. 
Despite a lot of efforts on this problem and the design of specific tools (so-called expander decompositions), the best algorithm known uses a large number of rounds, $\Theta(n^{1/3})$, and is randomized~\cite{ChangPSZ21}. The suspected hardness of this problem suggests that the certification problem could also be difficult (that is, could require large certificates), because insights and concrete lower bounds have been successfully transferred from the CONGEST model to local certification in the past~\cite{Censor-HillelPP20, FeuilloleyFHPP21}.  

\subsection{Restricting the class of graphs}

The previous subsection shows that there is basically not much hope for better results than Lemma~\ref{lem:existsFO} in general. Following the steps of other well-known metatheorems, the natural approach is then to restrict the class of graphs we consider. A first direction would be to consider graphs with bounded treewidth, as an attempt to obtain a Courcelle-like theorem for local certification. In other words, is it possible to certify any MSO formula on bounded treewidth graphs? We unfortunately did not find the answer to this question. Even worse, it is an open problem to determine if we can certify that the graph itself has treewidth at most $k$. It is problematic since, even if we can certify a property on bounded treewidth graphs, we would like to be sure that we can efficiently check if we are in such a graph class.

\begin{question}\label{q:treewidth}
Can we certify with $O(\log n)$ bits that a graph $G$ has treewidth at most $t$?
\end{question}

The main issue with Question~\ref{q:treewidth} is that, in order to certify that a graph has treewidth at most~$t$, one has to certify a distribution of the vertices in \emph{bags} of size at most $t+1$. However, there is no a priori useful bound on how many bags there are, nor how far in the graph can be the vertices of the same bag. Thus, we might need to propagate some important information at a large distance, which becomes an issue if we have too many bags. Even for some restrictions of treewidth which are apparently simpler, such as pathwidth, certifying them with $O(\log n)$ bits is still open.

Motivated by the unavoidable non-elementary dependence in the formula in Courcelle's theorem~\cite{Frick04}, Gajarsk\'y and Hlin\v en\'y~\cite{GajarskyH15} designed a linear-time FPT algorithm for MSO-model checking with elementary dependency in the sentence, by paying the price of considering a smaller class of graphs, namely graphs of bounded treedepth. Their result is essentially the best possible as shown soon after in~\cite{Lampis13}. 

For our purposes, the bounded treedepth perspective seems definitely more promising. Indeed, as a first clue, we prove that one can locally check that a graph has treedepth at most $t$ with logarithmic-size certificates. 

\begin{restatable}{theorem}{thmCertifyTreedepth}
\label{thm:certify-treedepth}
If $G$ is a graph of treedepth at most~$t$, then we can certify that $G$ has treedepth at most~$t$ with $O(t \log n)$ bits.
\end{restatable}

The proof of this result, that will appear in Section~\ref{sec:treedepth-certif}, consists in showing that we can efficiently simulate a tree decomposition of the input graph, and heavily uses that spanning trees and paths can be easily certified. The key point of the proof is the fact that the set of ancestors of a node $v$ is a separator between $v$ and the rest of the graph. This small separator allows to only have to deal with the subtrees plus a bounded size separator, which makes the certification easier.

The next problem in line is then MSO-model checking for graphs of bounded treedepth. In such classes, it happens that MSO and FO have the same expressive power~\cite{ElberfeldGT16}: for every $t$ and every MSO sentence, there exists a FO sentence satisfied by the same graphs of treedepth at most $t$. The second part of this article consists in showing our main result, summarized in the following statement.

\begin{theorem}
\label{thm:main}
Every FO (and hence MSO) sentence $\varphi$ can be locally certified with $O(t\log n+ f(t,\varphi))$-bit certificates on graphs of treedepth at most $t$. 
\end{theorem}

The proof of this result goes through two steps. The first one is a kernelization result for FO-model checking for graphs of bounded treedepth (Section~\ref{sec:kernel}). 

\begin{theorem}
\label{thm:kerneltd}
For every integer $k$ and every graph $G$, there exists a graph $H$, called the \emph{kernel}, whose size is bounded by a computable function of $k$ such that $G$ and $H$ satisfy the same set of FO sentences with at most $k$ quantifiers. 
\end{theorem}

Note that a kernelization result already exists for graphs of bounded shrubdepth~\cite{GajarskyH15}, which implies bounded treedepth. We however cannot use this result directly as a blackbox. Indeed, in order to prove the second part of our result, we need to certify locally that the kernel we computed is effectively the one computed from the original input graph. In particular, we need to provide a kernelization process that allows such a certification. That part fails if we use the kernel provided in~\cite{GajarskyH15} for graphs of bounded shrubdepth. Again, it is already not clear how to certify that a graph has bounded shrubdepth, which leads to the following question:

\begin{question}
Can graphs of shrub-depth at most $t$ be certified with $O(\log n)$ bits? Can we certify a kernel with $O(\log n)$ bits?
\end{question}

Informally, the main issue to generalize our result to shrubdepth is that vertices are only on leaves of the shrub-decomposition, and then we do not necessarily have small separators on which our proof is based for treedepth\footnote{Such small separators indeed do not necessarily exist since graphs of bounded shrubdepth might be dense.}. We think that solving this question will very likely allow to extend our proof to certification of MSO sentences for bounded shrubdepth graphs.

Now, given this kernel, we show how to certify its structure in Section~\ref{sec:kernel-certif}. And this is basically all we need for Theorem~\ref{thm:main}. Indeed, remember that the kernel size does not depend on the size of the original network $n$. Thus, it is possible to give the full map (that is, the adjacency matrix for example) of the kernel at every vertex, using only some $g(k,t)$ bits. Now, as the sentences satisfied by the kernel are the same as the sentences of the original graph, all vertices can just check independently in parallel that the kernel satisfies the sentence at hand.

\subsection{Related work and discussion of the model}

We refer to~\cite{Feuilloley21} for an introduction to local certification (both in terms of techniques and history), and to \cite{FeuilloleyF16} for a more complexity-theoretic survey. Most of the relevant related work has been reviewed in the previous sections, but there are two aspects that we would like to detail. 
First, there exists a line of work that builds bridges between distributed computation and logic, that we need to describe and compare to our approach. 
Second, we have already hinted that there is a difference between verification at distance one or at constant distance, and we elaborate a bit on this aspect.

\paragraph*{Distributed graph automata and modal logics}

A recent research direction consists in characterizing modal logics on graphs by various models of distributed local computation. This is similar to the approach of descriptive complexity in centralized computing, that aims at finding equivalences between computation models, and logic fragments (see~\cite{Immerman99} for a book on the topic). 

In this area, a paper that is especially relevant to us is~\cite{Reiter15}, which proves that MSO logic on graphs is equivalent to a model called alternating distributed graph automata. 
Let us describe what this automata model is, and then how it compares with our model. The nodes of the graph are finite-state machines, and they update their states synchronously in a constant number of  rounds. These are anonymous, that is, the nodes are not equipped with identifiers. The transition function of a node takes as input its state and the states of its neighbors in the form of a set (no counting is possible). 
At the end of the computation, the set of the states of the nodes, $F$, is considered, and the computation accepts if and only if $F$ is one of the accepting sets of states.
The alternating aspect is described in \cite{Reiter15} with computation branches, but in the context of our work it is more relevant to describe it informally as a prover/disprover game. The transition functions actually do not depend only on the states of the neighborhood, they also depend on additional labels given by two oracles, called prover and disprover. The prover and the disprover alternate in providing constant-size labels to the nodes, in order to reach respectively acceptance and rejection.

There are several substantial differences between our model and the model of~\cite{Reiter15}. First, our model is stronger in terms of local computation: we assume unbounded computation time, and non-constant space, whereas \cite{Reiter15} assumes finite-state machines. 
Second, our acceptance mechanism is weaker, in the sense that it is essentially the conjunction of a set of binary decisions, whereas \cite{Reiter15} uses an arbitrary function of a set of outputs. 
Third, we only have one prover, whereas \cite{Reiter15} has the full power of alternating oracles. 
Actually, variants of local certification using these two extensions have been considered (certification enhanced with general accepting functions in \cite{ArfaouiFIM14, ArfaouiFP13}, and generalized to an analogue of the polynomial hierarchy in \cite{FeuilloleyFH21, BalliuDFO18}), but here we are interested in the classic setting.

\paragraph*{Verification radius: one or constant}
An aspect of the model that is important in this paper is the locality of the verification algorithm. 
The original papers on local certification consider a model called \emph{proof-labeling schemes}~\cite{KormanKP10}, where the nodes only see their neighbors (in the spirit of the state model in self-stabilization~\cite{Dolev2000}).
Then, it was generalized in  \cite{GoosS16} to \emph{locally checkable proofs} where the vertices can look at a constant distance. They proved that several lower bounds (\emph{e.g.} for acyclicity certification) still hold in this model. 

The two models have pros and cons. Choosing constant distance is more appealing from a theoretical point of view, as it removes the distance $1$ constraint, which could seem arbitrary, but still captures a notion of locality. 
On the other hand, constant distance is not well-suited to contexts where we care about message sizes: with unbounded degree, looking at constant distance can translate into huge messages.
As noted in \cite{GoosS16}, due to their locality, FO formulas can be checked with no certificate if we can adapt the view of the node to the formula, and this can be extended to certification of monadic $\Sigma_1^1$ formulas if one allows $O(\log n)$-bit certificates. 

For this paper, we chose to fix the distance to $1$, in order to prevent this adaptation of the radius to the formula. 
Note that the difference between the two models can be dramatic. For example, deciding whether a graph has diameter $3$ or more, does not need any certificate if the nodes can see at distance $3$, but requires certificates of size linear in $n$ if they can only see their neighbors~\cite{Censor-HillelPP20}.

\section{Preliminaries}

\subsection{Treedepth}
Treedepth was introduced by Ne\v{s}et\v{r}il and Ossana de Mendez in~\cite{NesetrilM06} as a new graph parameter where model checking is more efficient. In the last ten years, this graph parameter received considerable attention (see \cite{NesetrilM12} for a book chapter about this parameter). Treedepth is related to other important width parameters in graphs. In particular, it is an upper bound on the pathwidth, which is another important parameter, especially in the study of minors~\cite{RobertsonS83} and interval graphs~\cite{Bodlaender98}.

Let $T$ be a rooted tree. A vertex $u$ is an \emph{ancestor} of $v$ in $T$, if $u$ is on the path between $v$ and the root. We say that $v$ is a \emph{descendant} of $u$ if $u$ is an ancestor of $v$.

\begin{definition}[\cite{NesetrilM06}]
The \emph{treedepth} of a graph $G$ is the minimum height of a forest $F$ on the same vertex set as $G$, such that for every edge $(u,v)$ of the graph $G$, $u$ is an ancestor or a descendant of $v$ in the forest. 
\end{definition}

When $G$ is connected, which is the case for this paper, the forest $F$ is necessarily a tree. 
Such a tree $T$ is called the \emph{elimination tree}, and, in a more logic-oriented perspective, it is called a \emph{model} of the graph. If the tree has depth at most $k$, it is a \emph{$k$-model of $G$} (see Figure~\ref{fig:treedepth}).
Note that there might be several such elimination trees / models. 

\begin{figure}[!h]
    \centering
    \begin{tabular}{cc}
    \scalebox{0.9}{
    \input{treedepth-1.tex}}
    &
    \scalebox{0.9}{
    \input{treedepth-2.tex}}
    \end{tabular}
    \caption{An example of an elimination tree. On the left the graph $G$, that is a path on seven vertices, and on the right an elimination tree $T$ of this graph. Since this tree has depth 2, the path has treedepth at most 2, and this is actually optimal.}
    \label{fig:treedepth}
\end{figure}
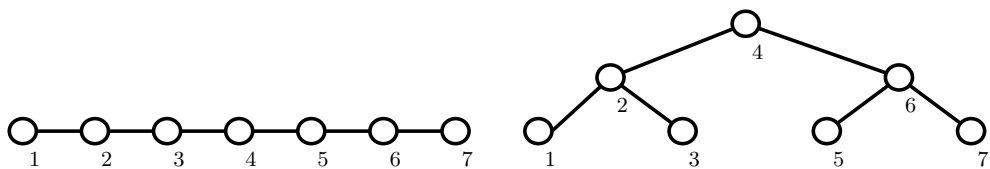

Let us fix an elimination tree. A \emph{vertex of $G$ has depth $d$}, if it has depth $d$ in the elimination tree. For any vertex $v$, let $G_v$ be the subgraph of $G$ induced by the vertices in the subtree of $T$ rooted in~$v$. Note that, for the root~$r$, $G_r=G$. Now, a model $T$ of $G$ is \emph{coherent} if, for every vertex $v$, the vertices of the subforest rooted in $v$ form a connected component in $G$. In other words, for every child $w$ of $v$, there exists a vertex $x$ of the subtree rooted in $w$ that is connected to~$v$.

One can easily remark that the following holds:

\begin{lemma}\label{lem:coherent1}
Let $G$ be a connected graph of treedepth $d$. Then there exists a tree $T$ that is a coherent $d$-model of $G$.
\end{lemma}
\begin{proof}
Let $T$ be a $d$-model of $G$ where the sum over all the vertices of $V$ of the depth of $v$ is minimized. We claim that $T$ is coherent. Assume by contradiction that there exists a vertex $v$, and one of its children $w$, such that no vertex of the subtree rooted in $w$ is connected to $v$. Let $v'$ be the lowest ancestor of $v$ connected to a vertex of $G_w$ (such a vertex must exist since $G$ is connected). We can attach the subtree of $w$ on $v'$ rather than $v$, without breaking the fact that the tree is a model of $G$. This new tree has a lower sum of depths than the original one, a contradiction with the minimality.
\end{proof}

One can wonder if we can assume that $w$ is connected to its closest ancestor. The answer is negative, for instance on the representation for a $k$-model of a path $P_{2^k-1}$. See Fig.~\ref{fig:treedepth}.

Using Lemma~\ref{lem:coherent1}, one can easily check that the following holds.

\begin{remark}
Let $T$ be a coherent $d$-model of a connected graph $G$ and $u$ be a vertex of $G$. Then $G_u$ induces a connected subgraph.
\end{remark}

\subsection{FO and MSO logics}

Graphs can be seen as relational structures on which properties can be expressed using logical sentences. The most natural formalism considers a binary predicate that tests the adjacency between two vertices. Allowing standard boolean operations and \emph{quantification on vertices}, we obtain the \emph{first-order logic} (FO for short) on graphs. Formally, a FO formula is defined by the following grammar:
$$x=y \mid x-y \mid \neg F \mid F\wedge F \mid F\vee F  \mid \forall x F \mid \exists x F$$
where $x,y$ lie in a fixed set of variables. Except for $x-y$, which denotes the fact that $x$ and $y$ are adjacent, the semantic is the classic one. 
Given a FO sentence $F$ (i.e. a formula where each variable falls under the scope of a corresponding quantifier) and a graph $G$, we denote by $G\vDash F$ when the graph $G$ satisfies the sentence $F$, which is defined in the natural way. 

MSO logic is an enrichment of FO, where we allow \emph{quantification on sets of vertices or edges}. We skip the formal definition of MSO here, because of the following result. It shows that for our purposes, FO and MSO have the same expressive power since we consider only bounded treedepth graphs. 

\begin{theorem}[\cite{Grohe17}]
For every integer $d$ and MSO sentence $\varphi$, there exists a FO sentence $\psi$ such that $\varphi$ and $\psi$ are satisfied by the same set of graphs of treedepth at most $d$. 
\end{theorem}

In the following, we are looking for a kernelization result for the model checking problem, where the kernel is checkable with small certificates. In particular, given a sentence $\varphi$ and a graph $G$, we have to prove that the graph $H$ output by our kernelization algorithm satisfies $\varphi$ if and only if so does $G$. We actually show a stronger result, namely that for every integer $k$ and every graph $G$, there exists a graph $H_k$ satisfying the same set of sentences with at most $k$ nested quantifiers as $G$. In that case we write $G\simeq_k H_k$. This yields the required result when $k$ is quantifier depth of $\varphi$. 

The canonical tool to prove equivalence between structures is the so-called Ehrenfeucht-Fraïssé game. This game takes place between two players, Spoiler and Duplicator. The arena is given by two structures (here, graphs) and a number $k$ of rounds. At each turn, Spoiler chooses a vertex in one of the graphs, and Duplicator has to answer by picking a vertex in the other graph. At turn $i$, assume that the positions played in the first (resp. second) graph are $u_1,\ldots,u_i$ (resp. $v_1,\ldots,v_i$). Spoiler wins at turn $i$ if the mapping $u_j\mapsto v_j$ is not an isomorphism between the subgraphs induced by $\{u_1,\ldots,u_i\}$ and $\{v_1,\ldots,v_i\}$. If Spoiler does not win before the end of the $k$-th turn, then Duplicator wins.

The main result about this game is the following, which relates winning strategies with equivalent structures for $\simeq_k$. 

\begin{theorem}
\label{thm:EF}
Let $G,H$ be two graphs and $k$ be an integer. Duplicator has a winning strategy in the $k$-round Ehrenfeucht-Fraïssé game on $(G,H)$ if and only if $G\simeq_k H$. 
\end{theorem}

See \cite{Thomas93} for a survey on Ehrenfeucht-Fraïssé games and its applications in computer science.

\subsection{Local certification: definitions and basic techniques}
\label{subsec:basic-certif}

We assume that the vertices of the graph are equipped with unique identifiers, also called IDs, in a polynomial range $[1, n^k]$ ($k$ being a constant). Note that an ID can be written on $O(\log n)$ bits.

In this paper, a local certification is described by a local verification algorithm, which is an algorithm that takes as input the identifiers and the labels of a node and of its neighbors, and outputs a binary decision, usually called \emph{accept} or \emph{reject}. A local certification of a logical formula is a local verification algorithm such that:
\begin{itemize}
    \item If the graph satisfies the formula, then there exists a label assignment, such that the local verification algorithm accepts at every vertex.
    \item If the graph does not satisfy the formula, then for every label assignment, there exists at least one vertex that rejects. 
\end{itemize}

A graph that satisfies the formula is a \emph{yes-instance}, and a graph that does not satisfy the formula is a \emph{no-instance}. 
The labels are called \emph{certificates}. It is equivalent to consider that there is an entity, called the \emph{prover}, assigning the labels (a kind of external oracle). The size $f(n)$ of a certification is the size of its largest label for graphs of size $n$. The certification size of a formula or a type of formula is (asymptotic) minimum size of a local certification. 

Let us now sketch some standard tools for local certification that we repeatedly use in the rest of the paper. (We refer to \cite{Feuilloley21} for full explanations and proofs.) Before going to certification, let us just remind that a vertex accesses its degree locally without any certificate. 

The most useful tool in local certification is arguably the spanning tree. Suppose you want to certify that at least one vertex of the graph has a special property (\emph{e.g.} at least one vertex has degree~$10$). 
On \emph{yes}-instances, each node has to know that somewhere in the graph, some vertex will indeed be special. To give this information to every node, the prover can design labels by choosing an arbitrary spanning tree of the graph rooted at the special vertex, and then provide each node with the following information: (1) the identifier of its parent in the tree (or itself, if it is the root), (2) the identifier of the root, and (3) the distance from the node to the root in the tree.

The vertices can check this certification. First, the vertices can perform basic sanity checks on the spanning structure, and this will be enough to ensure the correctness. For example, if the set of pointers contains a cycle, then there will be an inconsistency in the distances; And if it has several connected components, either two adjacent vertices will have been given different identifiers for the root, or one of the vertices with distance $0$ will not have the identifier announced for the root. 
Now that the spanning tree is certified, the root checks that it indeed has the special property (\emph{e.g.} being of degree $10$).

The spanning tree structure can also be used to certify the number of vertices in the graph: in addition to the spanning tree, on \emph{yes}-instances, the prover can give to every vertex the total number of nodes in the graph and the number of nodes in its subtree. Every node can locally check the sum, and the root can check that the two numbers it has been assigned are equal. 

Note that the certifications discussed above all use certificates on $O(\log n)$ bits.

\section{Certification of small fragments: Proof of Lemma~\ref{lem:existsFO}}
\label{sec:fragment}

This section is devoted to prove Lemma~\ref{lem:existsFO}.

\lemFragments*

Let us first prove the following lemma:

\begin{lemma}\label{lem:FOext}
Existential FO sentences with $k$ quantifiers (i.e. whose prenex normal form has only existential quantifiers) can be certified with $O(k \log n)$ bits.
\end{lemma}

\begin{proof}
Let $G$ be a connected graph and $\exists x_1\cdots\exists x_k\varphi$ be an existential FO sentence where $\varphi$ is quantifier-free. Let $v_1,\ldots,v_k$ be $k$ vertices such that the formula is satisfied by $v_1,\ldots,v_k$. 

Every vertex receives the following certificate:
\begin{itemize}
    \item The list of identifiers of vertices $v_1,\ldots,v_k$.
    \item The $k\times k$ adjacency matrix of the subgraph induced by $v_1,\ldots,v_k$.
    \item The certificate of a spanning tree rooted on $v_i$ for every $i \le k$ (see Subsection~\ref{subsec:basic-certif}).
\end{itemize}
Every node then checks the certificate as follows. First, every node checks that its neighbors have the same list of vertices $v_1,\ldots,v_k$ and the same adjacency matrix. Then every node checks the certificate of the spanning tree of each $v_i$. Finally, each of the vertices $v_1,\ldots,v_k$ can now use the adjacency matrix to evaluate $\varphi$ on $(v_1,\ldots,v_k)$ and check that it is satisfied.
\end{proof}

Let us now prove the second part of Lemma~\ref{lem:existsFO}.

\begin{lemma}
FO sentences with quantifier depth at most 2 can be certified with $O(\log n)$ bits.
\end{lemma}
\begin{proof}
First, observe that sentences with quantifier depth one are satisfied by either all graphs or none of them. We thus consider the depth 2 case.

Let $\varphi$ be a sentence of quantifier depth at most two. Without loss of generality, we may assume that $\varphi$ is a boolean combination of sentences of the form $Qx \psi(x)$ where $\psi(x)$ is again a boolean combination of formulas of the form $Qy\pi(x,y)$ where $\pi(x,y)$ is quantifier-free. 

Observe that up to semantic equivalence, $\pi(x,y)$ can only express that $x=y$, $xy$ is an edge, $xy$ is a non-edge, or the negation of these properties.

Trying the two possible ways of quantifying $y$ in these six properties, we end up showing (using that our graphs are connected) that $\psi(x)$ lies among these three properties or their negations:
\begin{itemize}
\item $x$ is the only vertex.
\item $x$ is a dominating vertex.
\item $x$ is not the only vertex but dominates the graph.
\end{itemize}

Now, quantifying on $x$ leaves only a few choices for $\varphi$, namely boolean combinations of the following:
\begin{enumerate}
    \item The graph has at most one vertex.
    \item The graph is a clique.
    \item The graph has a dominating vertex.
\end{enumerate}

Since certifying disjunction or conjunction of certifiable sentences without blow up (asymptotically) in size is straightforward, it is sufficient to show that the three properties and their negations can all be checked with $O(\log(n))$-bit certificates.

Since our graphs are connected, Property 1 is equivalent to say that every vertex has degree 0, which can be checked with empty certificates. Similarly, its negation is equivalent to having minimum degree 1 which can be checked similarly.

For Property 2 (resp. the negation of 3), we begin by computing the number $n$ of vertices in the graph and certify it locally (see Subsection~\ref{subsec:basic-certif}). The verification algorithm then just asks whether the degree of each vertex is $n-1$ (resp. less than $n-1$).

For Property 3 (resp. the negation of 2), we again compute and certify the number $n$ of vertices. We additionally certify a spanning tree rooted at a vertex of degree $n-1$ (resp. less than $n-1$). The root then just check that it has indeed the right degree.
\end{proof}

\section{Certification of treedepth}
\label{sec:treedepth-certif}

This section is devoted to the proof of the following theorem.

\thmCertifyTreedepth*

\begin{proof}
Let $v$ be a vertex, and $w$ be its parent in the tree, we define an \emph{exit vertex of $v$} as a vertex $u$ of $G_v$ connected to $w$. Note that such a vertex must exist, if the model is coherent. 

We now describe a certification.
On a \emph{yes}-instance, the prover finds a coherent elimination tree of depth at most $t$, and assigns the labels in the following way.

\begin{itemize}
    \item Every vertex $v$ is given the list of the identifiers of its ancestors, from its own identifier to the identifier of the root.
    \item For every vertex $v$, except the root, the prover describes and certifies a spanning tree of $G_v$, pointing to the exit vertex of $v$. (See Subsection~\ref{subsec:basic-certif} for the certification of spanning trees.) The vertices of the spanning tree are also given the depth $k$ of $v$ in the elimination tree.
\end{itemize}

Note that the length of the lists is upper bounded by $t$, and that every vertex holds a piece of spanning tree certification only for the vertices of its list, therefore the certificates are on $O(t \log n)$ bits. Now, the local verification algorithm is the following. For every vertex $v$ with a list $L$ of length~$d+1$, check that:

\begin{enumerate}
    \item \label{item:at-most-t} $d\leq t$, and $L$ starts with the identifier of the vertex, and ends with the same identifier as in the lists of its neighbors in the graph.
    \item \label{item:edges-ancestors} The neighbors in $G$ have lists that are suffix or extension by prefix of $L$.
    \item \label{item:ST-for-each-depth} There are $d$ spanning trees described in the certificates.
    \item \label{item:ST-verif} For every $k\leq d$, for the spanning trees associated with depth $d$: 
    \begin{itemize}
        \item The tree certification is locally correct.
        \item The neighbors in the tree have lists with the same $(k+1)$-suffix.
        \item If the vertex is the root, then it has a neighbor whose list is the $k$-suffix of its own list. 
    \end{itemize}
\end{enumerate}

It is easy to check that on \emph{yes}-instances the verification goes through.
Now, consider an instance where all vertices accept. We shall prove that then we can define a forest, such that the lists of identifiers given to the nodes are indeed the identifiers of the ancestors in this forest. Once this is done, the fact that Steps~\ref{item:at-most-t} and~\ref{item:edges-ancestors} accept implies that the forest is a tree of the announced depth, and is a model of the graph. Let us first prove the following claim:

\begin{claim}
\label{clm:one-step}
For every vertex $u$, with a list $L$ of size at least two, there exists another vertex $v$ in the graph whose list is the same as $L$ but without the first element. 
\end{claim}

Consider a vertex $u$ like in Claim~\ref{clm:one-step}, at some depth $d$. If all vertices accept, then this vertex is has a spanning tree corresponding to depth $d$ (by Step~\ref{item:ST-for-each-depth}), where all vertices have the same $(d+1)$-suffix, and the root of this tree has a neighbor whose list is $L$, without the first identifier, by Step~\ref{item:ST-verif}. This vertex is the $v$ of the claim. 

The claim implies that the whole tree structure is correct. Indeed, if we take the vertex set of $G$, and add a pointer from every vertex $u$ to its associated vertex $v$ (with the notations of the claim), then the set of pointers must form a forest. In particular, there cannot be cycles, because the size of the list is decremented at each step. Also, if the ancestors are consistent at every node, then they are consistent globally. This finishes the proof of Theorem~\ref{thm:certify-treedepth}.
\end{proof}

\section{A kernel for FO model checking}
\label{sec:kernel}

In this section, we describe our kernel for FO model checking in bounded treedepth graphs. This basically mean that we describe a way to associate to every graph, a smaller graph, that satisfies exactly the same formulas. Let us be a bit more precise.
For any graph $G$ (without treedepth assumption), and any integer $k$, there exists a graph $H_k$ whose size is at most $f(k)$, such that for any FO-property~$\varphi$ of quantifier depth at most $k$, $G\vDash\varphi$ if and only if $H_k\vDash\varphi$. 
Indeed, since we have a bounded number of formulas of quantifier depth at most $k$ (up to semantic equivalence), we have a bounded number of equivalent classes of graphs for $\simeq_k$. We can associate to each class the smallest graph of the class, whose size is indeed bounded by a function of only $k$.
However, this definition of $H_k$ is not constructive, which makes it impossible to manipulate for certification. 
What we do in this section is to provide a constructive way of finding a graph $H_k$ such that $G\simeq H_k$,  when $G$ has bounded treedepth. 
Moreover, still for bounded treedepth, we also need to make sure we can label the nodes of $G$ to encode locally~$H_k$, and certify locally the correctness of $H_k$.

\subsection{Construction}

Let $G$ be a graph of treedepth at most $t$, and let $k$ be an integer. Let $\mathcal{T}$ be a $t$-model of $G$. Let $v$ be a vertex of depth $i$ in the decomposition. 
We define the \emph{ancestor vector of $v$} as the $\{0,1\}$-vector of size $i$, where the $j$-th coordinate is 1, if and only if, $v$ is connected in $G$ to its ancestor at depth $j$. 

We can now define the \emph{type of a vertex $v$} as the subtree rooted on $v$ where all the nodes of the subtree are labelled with their ancestor vector. Note that in this construction, the ID of the nodes do not appear, hence several nodes might have the same type while being at completely different places in the graph or the tree.

Let us now define a subgraph of $G$ that we will call the \emph{$k$-reduced graph}.
If a node $u$ has more than $k$ children of the same type, a \emph{valid pruning} operation consists in removing the subtree rooted on one of these children. Note that in doing so, we change the structure of the subtrees of $u$ and of its ancestors, thus we also update their types. 
A \emph{$k$-reduced graph} $H$ of $G$ is a graph obtained from~$G$ by iteratively applying valid pruning operations on a vertex of the largest possible depth in $\mathcal{T}$ while it is possible. A vertex $v$ is \emph{pruned} for a valid pruning sequence if it is the root of a subtree that is pruned in the sequence. Note that there are some vertices of $G \setminus H$ that have been deleted, but that are not pruned. 

Let $G$ be a graph, and $H$ be a $k$-reduction of $G$.
We now define an \emph{end type} for every vertex of~$G$. The \emph{end type} (with respect to $H$\footnote{One can prove that it actually does not depend on $H$ but we do not need it in our proof.}) of a deleted vertex is the last type it has had, that is, its type in the graph $G'$ which is the current graph when it was deleted.
The \emph{end type} of a vertex $u$ of $H$ (that is, that has not been deleted) is simply its type in $H$.

Since we apply pruning operations on a vertex of the largest possible depth, if at some point we remove a vertex of depth $i$, then we never remove a subtree rooted on a vertex of depth strictly larger than $i$ afterwards. It implies that when a vertex at depth $i$ is deleted, the type of each node at depth at least $i$ is its end type. In particular, we get the following:

\begin{lemma}\label{lem:sametype}
Let $G$ be a graph and $H$ be a $k$-reduced graph of $G$. Let $u \notin H$ and $v \in H$, such that $u$ is a child of $v$. Then there exists exactly $k$ children of $v$ in $H$ whose end type is the end type of~$u$.
\end{lemma}
\begin{proof}
By assumption, it cannot be more than $k$ since otherwise one of the children of $v$ would have been deleted. Moreover, since $u$ is deleted but not $v$, then $u$ is the root of a subtree we deleted while pruning $v$. In particular, $u$ has at least $k$ siblings with the same type. Now since all these siblings have the same depth as $u$, their type when $u$ is deleted is their end type. To conclude, observe that by construction, at least $k$ such siblings lie in $H$ since we delete some only if at least $k$ others remain.
\end{proof}

\begin{lemma}\label{lem:kreduced_size}
The number of possible 
end types of a node at depth $d$ in a $k$-reduced graph of treedepth at most $t$ is bounded by $f_d(k,t)$.
\end{lemma}
\begin{proof}
Let us prove Lemma~\ref{lem:kreduced_size} and define $f_d$ by backward induction on $d$. 

We start with $d=t$. Since the $t$-model has depth $t$, the tree rooted on a vertex of depth $t$ should be a single vertex graph. So the set of different possible types at depth $t$ only depends on the edges between the vertex of depth $t$ and its ancestors. There are $f_t(k,t)=2^{t}$ such types. 

Now let us assume that the conclusion holds for nodes at depth $d+1$, and let us prove it for depth $d$.  Let $u$ be a vertex of depth $d$ and $v_1,\ldots,v_r$ be its children in the elimination tree. 
Since $u$ is a vertex of a $k$-reduced graph, at most $k$ children of $u$ have the same end type $T$ and, by induction, there are at most $f_{d+1}(k,t)$ end types of nodes at depth $d+1$.
So the end type of $u$ is determined by its neighbors in its list of ancestors (which gives $2^d$ choices) and the multiset of types of its children. Since $u$ has at most $k$ children of each type, the type of $u$ can be represented as a vector of length $f_{d+1}(k,t)$, where each coordinate has an integral value between $0$ and $k$. So there are at most $f_d(k,t):=2^d \cdot (k+1)^{f_{d+1}(k,t)}$ types of nodes at depth $d$.
\end{proof}

We will also need the following remark, which follows from the definition of type:

\begin{remark}\label{rmk:deftype}
Let $v$ be a vertex of $G$. The end type of $v$ can be deduced from: 
\begin{itemize}
    \item the adjacency of $v$ with its ancestors,
    \item the number of children of $v$ of end type $T$ for any possible end type $T$
\end{itemize} 
\end{remark}

\subsection{It is a kernel}

Let $G$ be a graph of treedepth $t$, $\mathcal{T}$ be a $t$-model of $G$, and $G'$ be a $k$-reduced graph of $G$. Observe that $G'$ is a subgraph of $G$, and denote by $\mathcal{T}'$ the restriction of $\mathcal{T}$ to the vertices of $G'$.

If $S\subset V(G)$, we denote by $\mathcal{T}_S$ the subtree of $\mathcal{T}$ induced by the vertices of $S$ and their ancestors. In particular, $\mathcal{T}'=\mathcal{T}_{V(G')}$. Moreover, two rooted trees are said to be \emph{equivalent} if there is an end type-preserving isomorphism between them.

The goal of this section is to prove that $G\simeq_k G'$. By Theorem~\ref{thm:EF}, this is equivalent to finding a winning strategy for Duplicator in the Ehrenfeucht-Fraissé game on $G,G'$ in $k$ rounds. To this end, we prove that she can play by preserving the following invariant.

\begin{claim}
Let $x_1,\ldots,x_{i}$ (resp. $y_1,\ldots,y_{i}$) be the positions played in $G$ (resp. $G'$) at the end of the $i$-th turn. Then the rooted trees $\mathcal{T}_{\{x_1,\ldots,x_i\}}$ and $\mathcal{T}'_{\{y_1,\ldots,y_i\}}$ are equivalent. 
\end{claim}

\begin{proof}
The invariant holds for $i=0$ since the two trees are empty. Assume now that it is true for some $i<k$. We consider the case where Spoiler plays on vertex $x_{i+1}$ in $G$, the other case being similar (and easier). 
Consider the shortest path in $\mathcal{T}_{\{x_1,\ldots,x_{i+1}\}}$ between $x_{i+1}$ and a vertex of $\mathcal{T}_{\{x_1,\ldots,x_{i}\}}$. We call this path $u_1,...,u_p$, with $u_1$ a node of $\mathcal{T}_{\{x_1,\ldots,x_{i}\}}$ and $u_p=x_{i+1}$. Note that, necessarily, for all $j\in [1,i]$, $u_j$ is the parent of $u_{j+1}$ in the tree. 

For $j=1,\ldots,p$, we will find a vertex $u'_j$ in $G'$ such that  $\mathcal{T}_{\{x_1,\ldots,x_{i},u_j\}}$ is equivalent to $\mathcal{T}'_{\{y_1,\ldots,y_{i},u'_j\}}$ (note that this implies that $u_j$ and $u'_j$ have the same end type). 

For $j=1$, first observe that $\mathcal{T}_{\{x_1,\ldots,x_{i},u_1\}}=\mathcal{T}_{\{x_1,\ldots,x_{i}\}}$, because $u_1$ belongs to $\mathcal{T}_{\{x_1,\ldots,x_{i}\}}$. Then, since $\mathcal{T}_{\{x_1,\ldots,x_{i}\}}$ is equivalent to $\mathcal{T}'_{\{y_1,\ldots,y_{i}\}}$, we can define $u'_1$ as the copy of $u_1$ in $\mathcal{T}'_{\{y_1,\ldots,y_{i}\}}$.

Assume now that $u'_1,\ldots,u'_j$ are constructed. 
Let $T$ be the end type of $u_{j+1}$ in $G$, and $r$ be the number of children of $u_j$ having $T$ as their end type (including $u_{j+1}$). 
By construction of $G'$ and $u'_j$, we know that $u'_j$ has $\min(r,k)$ children with type $T$ in $\mathcal{T'}$. 
Observe that at most $\min(r-1,i)$ children of $u_j$ of type $T$ in $\mathcal{T}$ can lie in $\mathcal{T}_{\{x_1,\ldots,x_i\}}$. 
Indeed, since $u_{j+1}$ does not belong to $\mathcal{T}_{\{x_1,\ldots,x_i\}}$, we get the $r-1$ term, and since $\mathcal{T}_{\{x_1,\ldots,x_i\}}$ is made by $i$ vertices and their ancestors, not more than $i$ vertices of $\mathcal{T}_{\{x_1,\ldots,x_i\}}$ can have the same parent.
Also, using $i<k$, we get $\min(r-1,i)\leqslant \min(r,k)-1$. Therefore, there exists a child $u'_{j+1}$ of $u'_j$ of type $T$ in $\mathcal{T}'\setminus\mathcal{T}'_{\{y_1,\ldots,y_i\}}$. 

By taking $y_{i+1}=u'_p$, we finally obtain that  $\mathcal{T}_{\{x_1,\ldots,x_{i},u_p\}}=\mathcal{T}_{\{x_1,\ldots,x_{i+1}\}}$ is equivalent to $\mathcal{T}'_{\{y_1,\ldots,y_{i},u'_p\}}=\mathcal{T}'_{\{y_1,\ldots,y_{i+1}\}}$, as required. 
\end{proof}

\section{Certification of the kernel}
\label{sec:kernel-certif}

Now, let us prove that we can certify a $k$-reduction of $G$ obtained by a valid pruning of $G$. First obtain a $k$-reduced graph $H$ of $G$ via a valid pruning. Note that by Lemma~\ref{lem:sametype} if a vertex $u$ is in $G \setminus H$ then there are exactly $k$ other children of $u$ of the same end type. Moreover, by Lemma~\ref{lem:kreduced_size}, there exists a bounded number $f_d(k,y)$ of end types of vertices of depth $d$ in any $k$-reduced graph of treedepth at most $t$. So we can label all these end types with a label of size at most $\log(f_d(k,t))$. 

\begin{lemma}
Let $k$ be an integer. Let $G$ be a graph of treedepth at most $t$ with a coherent model $\mathcal{T}$. Let $H$ be a $k$-reduction of $G$ obtained via a valid pruning from $\mathcal{T}$.
Then we can certify with certificates of size $O(t\log n+g(k,t))$ that $H$ is a $k$-reduction of $G$ from $\mathcal{T}$.
\end{lemma}
\begin{proof}
Let us describe a local certification. On a \emph{yes}-instance, the prover gives to every vertex $v$ the following certificate:
\begin{itemize}
    \item The $O(t\log n)$-bit certificate of $v$ for the $t$-model $\mathcal{T}$ of $G$ given in Theorem~\ref{thm:certify-treedepth}.
    \item A list of booleans that says, for any ancestor $x$ of $v$ including $v$ if $x$ is pruned, i.e. the subtree rooted on $x$ has been pruned at some step.
    \item For every ancestor $w$ of $v$ including $v$, the end type of $w$.
\end{itemize}
Every node $v$ at depth $d$ thus receives a certificate of size at most $O(t\log n+d+\sum_{i=1}^d\log (f_i(k,t)))$. Let us now describe the local verification algorithm, as well as why it is sufficient for checkability.

By Remark~\ref{rmk:deftype}, the end type of a vertex only depends on its adjacency with its list of ancestors as well as the end type of its children. 
So first, the node $v$ can check that its adjacency with its list of ancestors is compatible with its end type. 
Then, it checks that, if one of its children $w$ has been pruned, then it has exactly $k$ children with the type of $w$ that have not been pruned (there is no type $T$ such that more than $k$ children of type $T$ are left after pruning).
Note that $v$ has access to all this information since, for every child $w$, there is a vertex $x$ in the subtree rooted on $w$ adjacent to $v$, because $\mathcal{T}$ is coherent.
Finally, since the end type of $v$ is determined by the end types of its children, $v$ simply has to check that its end type is consistent with the list of end types of its children.

As in the proof of Theorem~\ref{thm:certify-treedepth}, for any child $w$ of $v$, if the prover has cheated and the type of $w$ has been modified between $w$ and the exit vertex of $w$, then one node of the path from $w$ to the exit vertex should discover it, which ensures that the certification is correct. 
\end{proof}

\bibliography{MSO-certif}


\end{document}

%% file: treedepth-1.tex
\tikzset{every picture/.style={line width=0.75pt}} 

\begin{tikzpicture}[x=0.75pt,y=0.75pt,yscale=-1,xscale=1]

\draw [line width=1.5]    (100,137) -- (140,137) ;
\draw [line width=1.5]    (132.5,137) -- (172.5,137) ;
\draw [line width=1.5]    (172.5,137) -- (212.5,137) ;
\draw [line width=1.5]    (212.5,137) -- (252.5,137) ;
\draw [line width=1.5]    (252.5,137) -- (292.5,137) ;
\draw [line width=1.5]    (292.5,137) -- (332.5,137) ;
\draw  [fill={rgb, 255:red, 255; green, 255; blue, 255 }  ,fill opacity=1 ][line width=1.5]  (85,137) .. controls (85,133.13) and (88.36,130) .. (92.5,130) .. controls (96.64,130) and (100,133.13) .. (100,137) .. controls (100,140.87) and (96.64,144) .. (92.5,144) .. controls (88.36,144) and (85,140.87) .. (85,137) -- cycle ;
\draw  [fill={rgb, 255:red, 255; green, 255; blue, 255 }  ,fill opacity=1 ][line width=1.5]  (125,137) .. controls (125,133.13) and (128.36,130) .. (132.5,130) .. controls (136.64,130) and (140,133.13) .. (140,137) .. controls (140,140.87) and (136.64,144) .. (132.5,144) .. controls (128.36,144) and (125,140.87) .. (125,137) -- cycle ;
\draw  [fill={rgb, 255:red, 255; green, 255; blue, 255 }  ,fill opacity=1 ][line width=1.5]  (165,137) .. controls (165,133.13) and (168.36,130) .. (172.5,130) .. controls (176.64,130) and (180,133.13) .. (180,137) .. controls (180,140.87) and (176.64,144) .. (172.5,144) .. controls (168.36,144) and (165,140.87) .. (165,137) -- cycle ;
\draw  [fill={rgb, 255:red, 255; green, 255; blue, 255 }  ,fill opacity=1 ][line width=1.5]  (205,137) .. controls (205,133.13) and (208.36,130) .. (212.5,130) .. controls (216.64,130) and (220,133.13) .. (220,137) .. controls (220,140.87) and (216.64,144) .. (212.5,144) .. controls (208.36,144) and (205,140.87) .. (205,137) -- cycle ;
\draw  [fill={rgb, 255:red, 255; green, 255; blue, 255 }  ,fill opacity=1 ][line width=1.5]  (245,137) .. controls (245,133.13) and (248.36,130) .. (252.5,130) .. controls (256.64,130) and (260,133.13) .. (260,137) .. controls (260,140.87) and (256.64,144) .. (252.5,144) .. controls (248.36,144) and (245,140.87) .. (245,137) -- cycle ;
\draw  [fill={rgb, 255:red, 255; green, 255; blue, 255 }  ,fill opacity=1 ][line width=1.5]  (285,137) .. controls (285,133.13) and (288.36,130) .. (292.5,130) .. controls (296.64,130) and (300,133.13) .. (300,137) .. controls (300,140.87) and (296.64,144) .. (292.5,144) .. controls (288.36,144) and (285,140.87) .. (285,137) -- cycle ;
\draw  [fill={rgb, 255:red, 255; green, 255; blue, 255 }  ,fill opacity=1 ][line width=1.5]  (325,137) .. controls (325,133.13) and (328.36,130) .. (332.5,130) .. controls (336.64,130) and (340,133.13) .. (340,137) .. controls (340,140.87) and (336.64,144) .. (332.5,144) .. controls (328.36,144) and (325,140.87) .. (325,137) -- cycle ;

\draw (94.5,147.4) node [anchor=north west][inner sep=0.75pt]    {$1$};
\draw (134.5,147.4) node [anchor=north west][inner sep=0.75pt]    {$2$};
\draw (174.5,147.4) node [anchor=north west][inner sep=0.75pt]    {$3$};
\draw (214.5,147.4) node [anchor=north west][inner sep=0.75pt]    {$4$};
\draw (254.5,147.4) node [anchor=north west][inner sep=0.75pt]    {$5$};
\draw (294.5,147.4) node [anchor=north west][inner sep=0.75pt]    {$6$};
\draw (334.5,147.4) node [anchor=north west][inner sep=0.75pt]    {$7$};

\end{tikzpicture}

%% file: treedepth-2.tex
\tikzset{every picture/.style={line width=0.75pt}} 

\begin{tikzpicture}[x=0.75pt,y=0.75pt,yscale=-1,xscale=1]

\draw [line width=1.5]    (100,137) -- (132.5,107) ;
\draw [line width=1.5]    (132.5,107) -- (172.5,137) ;
\draw [line width=1.5]    (132.5,107) -- (207.5,77) ;
\draw [line width=1.5]    (292.5,107) -- (252.5,137) ;
\draw [line width=1.5]    (292.5,107) -- (207.5,77) ;
\draw [line width=1.5]    (292.5,107) -- (332.5,137) ;
\draw  [fill={rgb, 255:red, 255; green, 255; blue, 255 }  ,fill opacity=1 ][line width=1.5]  (85,137) .. controls (85,133.13) and (88.36,130) .. (92.5,130) .. controls (96.64,130) and (100,133.13) .. (100,137) .. controls (100,140.87) and (96.64,144) .. (92.5,144) .. controls (88.36,144) and (85,140.87) .. (85,137) -- cycle ;
\draw  [fill={rgb, 255:red, 255; green, 255; blue, 255 }  ,fill opacity=1 ][line width=1.5]  (125,107) .. controls (125,103.13) and (128.36,100) .. (132.5,100) .. controls (136.64,100) and (140,103.13) .. (140,107) .. controls (140,110.87) and (136.64,114) .. (132.5,114) .. controls (128.36,114) and (125,110.87) .. (125,107) -- cycle ;
\draw  [fill={rgb, 255:red, 255; green, 255; blue, 255 }  ,fill opacity=1 ][line width=1.5]  (165,137) .. controls (165,133.13) and (168.36,130) .. (172.5,130) .. controls (176.64,130) and (180,133.13) .. (180,137) .. controls (180,140.87) and (176.64,144) .. (172.5,144) .. controls (168.36,144) and (165,140.87) .. (165,137) -- cycle ;
\draw  [fill={rgb, 255:red, 255; green, 255; blue, 255 }  ,fill opacity=1 ][line width=1.5]  (200,77) .. controls (200,73.13) and (203.36,70) .. (207.5,70) .. controls (211.64,70) and (215,73.13) .. (215,77) .. controls (215,80.87) and (211.64,84) .. (207.5,84) .. controls (203.36,84) and (200,80.87) .. (200,77) -- cycle ;
\draw  [fill={rgb, 255:red, 255; green, 255; blue, 255 }  ,fill opacity=1 ][line width=1.5]  (245,137) .. controls (245,133.13) and (248.36,130) .. (252.5,130) .. controls (256.64,130) and (260,133.13) .. (260,137) .. controls (260,140.87) and (256.64,144) .. (252.5,144) .. controls (248.36,144) and (245,140.87) .. (245,137) -- cycle ;
\draw  [fill={rgb, 255:red, 255; green, 255; blue, 255 }  ,fill opacity=1 ][line width=1.5]  (285,107) .. controls (285,103.13) and (288.36,100) .. (292.5,100) .. controls (296.64,100) and (300,103.13) .. (300,107) .. controls (300,110.87) and (296.64,114) .. (292.5,114) .. controls (288.36,114) and (285,110.87) .. (285,107) -- cycle ;
\draw  [fill={rgb, 255:red, 255; green, 255; blue, 255 }  ,fill opacity=1 ][line width=1.5]  (325,137) .. controls (325,133.13) and (328.36,130) .. (332.5,130) .. controls (336.64,130) and (340,133.13) .. (340,137) .. controls (340,140.87) and (336.64,144) .. (332.5,144) .. controls (328.36,144) and (325,140.87) .. (325,137) -- cycle ;

\draw (94.5,147.4) node [anchor=north west][inner sep=0.75pt]    {$1$};
\draw (134.5,117.4) node [anchor=north west][inner sep=0.75pt]    {$2$};
\draw (174.5,147.4) node [anchor=north west][inner sep=0.75pt]    {$3$};
\draw (209.5,87.4) node [anchor=north west][inner sep=0.75pt]    {$4$};
\draw (254.5,147.4) node [anchor=north west][inner sep=0.75pt]    {$5$};
\draw (294.5,117.4) node [anchor=north west][inner sep=0.75pt]    {$6$};
\draw (334.5,147.4) node [anchor=north west][inner sep=0.75pt]    {$7$};

\end{tikzpicture}